\documentclass[10pt,conference,letterpaper]{IEEEtran}

\usepackage{amsmath,amssymb,amsthm}
\usepackage{booktabs}
\usepackage{graphicx}
\usepackage{xcolor}
\usepackage{url}
\usepackage{cite}
\usepackage{multirow}
\usepackage{array}
\usepackage{algorithm}
\usepackage{algpseudocode}
\usepackage{placeins}
\usepackage{stfloats}

\newcommand{\system}{VeraRAN}
\newcommand{\repair}{VeraSync}
\newcommand{\por}{MI-POR}
\newcommand{\Req}{\mathsf{REQ}}

\newcommand{\Apply}{\mathsf{APPLY}}

\newcommand{\Obs}{\mathsf{OBS}}
\newcommand{\Exec}{\mathsf{Exec}}

\newcommand{\Verify}{\mathsf{Verify}}
\newcommand{\Pred}{\operatorname{Pred}}
\newcolumntype{L}[1]{>{\raggedright\arraybackslash}p{#1}}

\newtheorem{definition}{Definition}
\newtheorem{lemma}{Lemma}
\newtheorem{theorem}{Theorem}

\title{VeraRAN: Pre-Actuation Certification and Event-Causal Synchronization Repair for
Asynchronous Multi-Interface RAN Plans}
\author{
\IEEEauthorblockN{Yinghan Hou and Zongyou Yang}
\IEEEauthorblockA{Imperial College London, United Kingdom
\enspace \{yh24, zy2926\}@ic.ac.uk}
}

\begin{document}
\bstctlcite{BSTcontrol}
\maketitle

\begin{abstract}
Agentic RAN controllers combine mobility, energy, and resource actions across
independently implemented interfaces. Even when each command is valid and the
target state is safe, asynchronous actuation can drive the network through
unsafe intermediate states. In a frozen study of a 35B planner, 28.8\% of
locally valid plans remained asynchronously unsafe. We introduce VeraRAN,
which checks plans before actuation by modeling request, delivery, acceptance,
application, completion, and observation for each action while exploring
plausible delays and event orders. When VeraRAN finds a counterexample,
VeraSync inserts versioned event barriers and rechecks the repaired plan for
safety and completion. MI-POR prunes independent interleavings using RAN
lifecycle and resource footprints. In a post-freeze stratified confirmation
within the declared repair domain, VeraSync re-certified every confirmation
plan while leaving 87\% of action pairs unordered. MI-POR matched exact search
in a property-stratified audit and reduced explored states by 94.6--95.0\% on
20--40-action plans. Native ns-O-RAN replay and an independent live E2 audit
showed why distinguishing these events matters: acceptance may precede the
authoritative state transition, so dependent actions must wait for direct
APPLY evidence or a contract-backed completion event causally downstream of
APPLY.
\end{abstract}


\section{Introduction}

A RAN plan may combine mobility, energy, scheduling, load, and policy
controls~\cite{oh2011dynamic,li2011energyefficient,caballero2017multitenant}.
O-RAN exposes them through separately operated control and management paths
~\cite{garciasaavedra2021oran,polese2023understanding,hoffmann2024xapps,
oran2023e2ap,etsi2024a1,etsi2024o1}. Emerging intent-driven and agentic
controllers motivate coordinated control across RAN functions and timescales
~\cite{bonati2021intelligence,elkael2026agentran}. Yet local validity does not
constrain the order in which such controls take effect. A source cell may
sleep before the serving-cell transition applies; a donor quota may fall
before its replacement becomes active, even though every command and the
final state are valid.

We model E2 control, A1 policy application, and O1 management as
independently contracted actions whose delivery envelopes and authoritative
events may differ. Local gates validate one command, fixed waits approximate
actuation with elapsed time, and global serialization orders related and
unrelated actions alike. These safeguards do not identify the authoritative
event that discharges a cross-path dependency. 

\system{} supplies a pre-actuation certification pipeline
(Fig.~\ref{fig:system-story}). A
\emph{per-action event contract} relates request, delivery, acceptance,
authoritative application, completion, and observation. The verifier explores
admissible endpoints, delays, orders, versions, and evidence scopes.
Counterexamples drive \repair{} to add event dependencies, after which the
entire revised plan undergoes the same safety-and-completion check. A
fail-closed executor releases each gated REQUEST, or opens a registered
application hook, only when authoritative evidence matches the required
action, scope, version, and epoch.

Across a balanced six-stratum study, 51 of 177 locally valid 35B plans admitted
an unsafe asynchronous execution. In an independent 144-plan confirmation
within the declared repair domain, \repair{} re-certified every plan. Initial
portable REQUEST-gate synthesis left 86.98\% of action pairs unordered;
registered application hooks enabled verified pruning to 89.33\%. Live E2
measurements grounded the release events, while the
multi-interface partial-order reduction (\por{}) reduced explored states by
94.6--95.0\% on decomposable 20--40-action plans.

\begin{figure*}[t]
  \centering
  \includegraphics[width=0.99\textwidth]{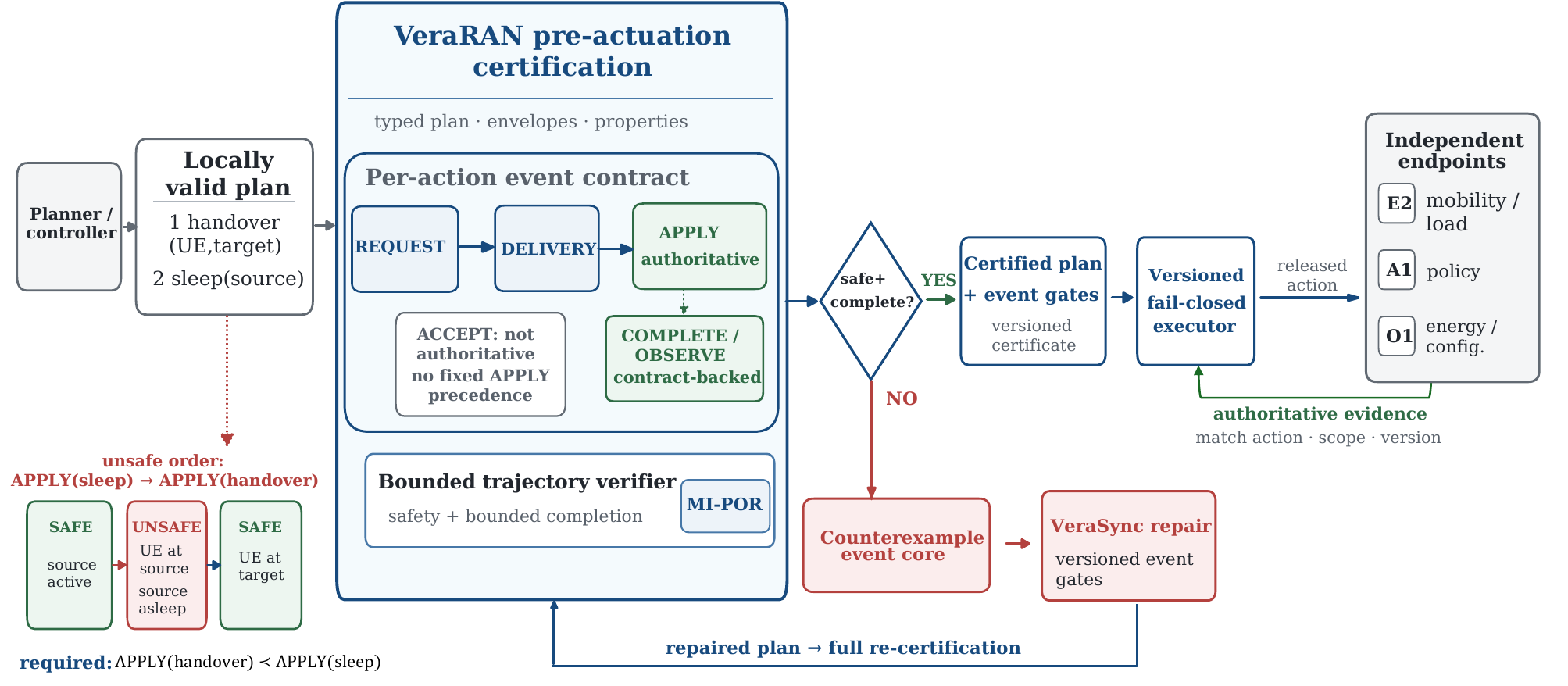}
  \caption{\textbf{\system{} architecture and event-causal execution.}
  The verifier certifies a candidate plan or sends its counterexample to
  \repair{} for repair with registered templates and full re-certification.
  The executor releases a gated REQUEST or opens a registered application
  hook only on authoritative evidence matching the certified action, scope,
  version, and epoch.}
  \label{fig:system-story}
\end{figure*}

The central contribution is an end-to-end refinement from a symbolic
cross-path dependency to the runtime evidence allowed to discharge it.
\repair{} instantiates registered fail-closed gates and rechecks the full plan;
\por{} removes redundant interleavings through lifecycle, evidence, resource,
repair, and necessary-enabling footprints without changing declared outcomes.

\section{Unsafe Plan Trajectories and Release Authority}

\subsection{Safe commands can form an unsafe trajectory}

For handover followed by source-cell sleep, the UE must move before its source
cell turns off. The required dependency is
\begin{multline}
\Apply(\mathrm{ServingCell}(u)=c_t,v_h)\\
\prec \Apply(\mathrm{CellSleep}(c_s),v_s).
\label{eq:handover-example}
\end{multline}
A quota migration transfers $\Delta$ from donor slice $d$ to recipient $r$.
The decrease can be issued only after the increase applies:
\[
\Apply(q_r\leftarrow q_r+\Delta,v_q)
\prec
\Req(q_d\leftarrow q_d-\Delta,v_q).
\]
The stronger APPLY-to-REQUEST form prevents an executor from delivering the
decrease early to an endpoint that could apply it immediately.

\subsection{Release authority}

Only APPLY changes authoritative state. ACCEPT cannot release a dependency by
label alone; COMPLETE or OBSERVE may release it only when the registered event
contract places that event downstream of the matching APPLY transition.
Accordingly, \system{} represents each dependency as a typed relation from an
authoritative predecessor event to the successor's REQUEST or APPLY gate.

\section{Plan Certification with Event Contracts}

\subsection{Typed plan}

A compiled plan is
\[
\mathcal P=(A,\prec_0,\Gamma,\mathcal D,\Phi,H).
\]
$A$ contains typed actions and $\prec_0$ their declared order. $\Gamma$
contains the event contracts and intrinsic lifecycle delays; $\mathcal D$
contains endpoint choices and root REQUEST-ready bounds.
$\Phi=\Phi_S\land\Phi_L$ gives the safety and completion properties, and $H$
is the horizon. Runtime evidence and versions belong to the verifier state,
not to $\mathcal P$. Let
\[
\mathcal V_a^{\mathrm{all}}=\{R_a,D_a,A_a,P_a,C_a,O_a\},
\]
denote the universe of possible lifecycle events for action $a$, and let
$V_a\subseteq\mathcal V_a^{\mathrm{all}}$ be the subset present in its
registered contract. The six symbols denote REQUEST, DELIVERY, ACCEPT,
authoritative APPLY, mechanism COMPLETE, and independent OBSERVE. The
per-action event contract is
\[
\Gamma_a=(V_a,F_a,\delta_a,\kappa_a,T_a),
\]
where $F_a$ is a causal DAG, $\delta_a$ bounds edge delays, $\kappa_a$
assigns evidence signatures, and
$T_a=\operatorname{Term}(a)\in V_a$ is the declared terminal event.
$V_a$ includes $R_a,D_a,P_a$; a contract without $C_a$ can terminate at
$P_a$ or $O_a$. The compiler requires
$R_a\prec_{\Gamma_a}^{*}z$ for every $z\in V_a\setminus\{R_a\}$, and rejects
an unreachable event or terminal.
Only $R_a\prec_{\Gamma_a}D_a\prec_{\Gamma_a}P_a$ is mandatory. No order
involving $A_a,C_a,O_a$ is assumed without a declared edge. Define
\[
\begin{aligned}
\operatorname{Auth}(\Gamma_a)
={}&\{P_a\}\cup\{z\in V_a\cap\{C_a,O_a\}:\\
&P_a\preceq_{\Gamma_a}^{*}z
\land\kappa_a(z)\equiv\kappa_a(P_a)\}.
\end{aligned}
\]
Here $\equiv$ permits a different wire type but requires the same action,
target scope, version, and epoch. ACCEPT is not authoritative by label; a
direct-callback contract may omit $A_a,C_a$. The audited OCUDU contract is
\[
R_a\prec_{\Gamma_a}D_a\prec_{\Gamma_a}A_a
\prec_{\Gamma_a}P_a\prec_{\Gamma_a}C_a,\qquad
P_a\prec_{\Gamma_a}O_a,
\]
sets $T_a=C_a$, and assumes no unmeasured $C_a/O_a$ order. Distinct events
remain strictly ordered in the DAG even when their logical timestamps agree;
for example, $\theta_a(A_a)\leq\theta_a(P_a)$ permits a legal equal-time
$A_a$-then-$P_a$ order.

\subsection{State and transition relation}

The verifier state is
\[
s=(x,q,\ell,\theta,\sigma,e,\nu,\tau),
\]
where $x,q$ are authoritative radio/resource state;
$\ell_a(z)\in\{\mathrm{pending},\mathrm{fired}\}$ and $\theta_a(z)$ record
lifecycle status and firing time;
$\sigma_a(z)\in\mathbb N_0\cup\{\bot\}$ is the scheduled firing time;
$e,\nu$ store runtime evidence and versions; and $\tau\in\mathbb N_0$ is
logical time. Let
$\Pred_a(z)=\{u:(u,z)\in F_a\}$.
\[
\begin{aligned}
\mathrm{Ready}_s(z)\Longleftrightarrow {}&
\Pred_a(z)\subseteq\mathrm{Fired}(s)\\[-1mm]
&{}\land G_{\rm plan}\land G_{\rm id}\land G_{\rm repair},\\
L_s(z)={}&\max_{u\in\Pred_a(z)}
       \{\theta_a(u)+d^-_{uz}\},\\
U_s(z)={}&\min_{u\in\Pred_a(z)}
       \{\theta_a(u)+d^+_{uz}\},\\
\mathrm{Enabled}_s(z)\Longleftrightarrow {}&
\ell_a(z)=\mathrm{pending}\land\mathrm{Ready}_s(z)\\[-1mm]
&{}\land\sigma_a(z)\neq\bot\land\sigma_a(z)\leq\tau.
\end{aligned}
\]
$G_{\rm plan}$ enforces plan order and endpoint eligibility, $G_{\rm id}$
matches action, scope, version, epoch, and freshness, and $G_{\rm repair}$
enforces instantiated barriers. Initially $\sigma_a(z)=\bot$ for non-root
events. At initialization, the verifier nondeterministically assigns each root
REQUEST an immutable scheduled time from its interval in $\mathcal D$. When a
transition first makes any other $z$ ready by firing its last causal
predecessor or discharging a guard, it assigns
$\sigma_a(z)\in[\max\{\tau,L_s(z)\},U_s(z)]$; an empty interval makes $z$
unrealizable. Firing $z$ sets $\ell_a(z)=\mathrm{fired}$ and
$\theta_a(z)=\tau$, applies its registered update, and schedules newly ready
successors. Let
\[
\begin{aligned}
E_\tau(s)&=\{z:\mathrm{Enabled}_s(z)\},\\[-1mm]
\mathrm{Tick}(s)&\Longleftrightarrow
\tau<H\land\neg\Phi_L(s)\land E_\tau(s)=\varnothing.
\end{aligned}
\]
A tick transition changes only logical time, setting $\tau\leftarrow\tau+1$.
Thus a due event cannot be postponed by TICK. The implementation coalesces
consecutive ticks by jumping to the next scheduled time. At a nonempty
$E_\tau(s)$, the verifier branches over every event and legal equal-time
order; an unrealizable event or an empty transition set before completion is
a completion failure.
REQUEST/DELIVERY carry an endpoint. ACCEPT, COMPLETE, and OBSERVE
record typed evidence and time, whereas APPLY alone changes authoritative
state and version. Later evidence cannot rewrite its APPLY time.

\paragraph{Exact six-event semantics}
The reference retains every event read by a property, terminal predicate,
timing rule, or gate, and enumerates the DAG, integer delays, and legal
equal-time orders. Release requires matching kind, identity, scope, version,
epoch, and freshness; deadlock cannot pass on safety alone.

$\Exec(\mathcal P)$ contains the maximal executions through $H$:
a trace stops at completion, deadlock, or after exhausting event transitions
at $\tau=H$. Its monotone completion predicate is
$\Phi_L(s)\equiv\bigwedge_{a\in A}
[\ell_a(T_a)=\mathrm{fired}]$; a completed trace has time
$T_{\rm comp}(\rho)=\max_{a\in A}\theta_a(T_a)$. A plan is certified iff
\begin{align}
\forall \rho\in\Exec(\mathcal P),\ \forall s\in\rho &: \Phi_S(s),
\label{eq:safety}\\
\forall \rho\in\Exec(\mathcal P),\ \exists s\in\rho &: \Phi_L(s).
\label{eq:liveness}
\end{align}
A reject-only deadlock may satisfy \eqref{eq:safety} but violates
\eqref{eq:liveness}, so it is not a valid certificate.

\subsection{Properties and event dependencies}

The two-action dependency in \eqref{eq:handover-example} generalizes to a
handover group $G$: every declared member needs a matching target-cell APPLY
before source sleep.
\begin{multline}
\forall u\in G:\quad
\Apply(\mathrm{ServingCell}(u)=c_t,v_h)\\
\prec \Apply(\mathrm{CellSleep}(c_s),v_s).
\label{eq:handover-group}
\end{multline}
The service-floor property is
\[
\phi_{\mathrm{floor}}(s)\equiv
\sum_{k\in\mathcal S}q_k(s)\ge Q_{\min}.
\]
The exact compiler instantiates this property for a bounded two-action quota
migration.

\begin{lemma}[Two-action service-floor barrier]
\label{lem:floor}
Assume an initial allocation with
$\sum_k q_k\ge Q_{\min}$, one recipient increase $+\Delta$, one donor decrease
$-\Delta$, no other mutation of these quotas during the migration, and
reliable execution of issued actions within the declared horizon. The barrier
\[
\Apply(q_r\leftarrow q_r+\Delta,v)
\prec\Req(q_d\leftarrow q_d-\Delta,v),
\]
preserves $\phi_{\mathrm{floor}}$ for every admitted delivery and application
delay.
\end{lemma}

\emph{Proof.} The sum is unchanged initially, rises by $\Delta$ after the
recipient APPLY, and returns to its initial value after the donor APPLY;
identity matching prevents another migration from releasing the decrease.
Lemma~\ref{lem:floor} covers only the declared two-action aggregate floor.
Per-slice bounds, nonnegativity, total capacity, and all other obligations
remain in $\Phi_S$ and are checked by the whole-plan verifier.

\section{\repair{}: From Counterexample to Executable Certificate}
\label{sec:verasync}

\subsection{Counterexample core}

From an unsafe trace, \repair{} extracts the violated property, authoritative
APPLY transitions, missing evidence, versions, and radio/resource footprints.
\repair{} repairs over the original actions and endpoints by adding
authoritative-event precedence relations. When no declared template expresses
the required relation, the executor withholds the plan and reports
\textsc{unsupported}.
Define
$B(b,a;z,G)$ to open gate $G\in\{\Req,\Apply\}$ of action $a$ only after
$z\in\operatorname{Auth}(\Gamma_b)$ fires with matching type, scope, and
version.
$\mathfrak B_{\mathcal T}(\mathcal P)$ denotes the finite set of barriers
that the registered template library $\mathcal T$ can instantiate over the
actions in $\mathcal P$.
Writing $\Apply$ or $\Obs$ as the release event abbreviates $z=P_b$ or
$z=O_b$. The verifier reasons about that event, not the label of a transport
ACK. We call $B$ the symbolic barrier and its executor check the runtime
gate.

\subsection{Template-guided repair}

Algorithm~\ref{alg:verasync} preserves actions, endpoints, and planner-specified
request-ready times; barriers may delay issuance. Each iteration compiles one
property-relevant core and rechecks the full plan. Here
$\mathcal P\oplus\mathcal B$ installs $\mathcal B$ in $\mathcal P$ and
$K_{\max}=4$.
The verifier returns
\[
\begin{aligned}
\Verify(\cdot)\in\{&
\textsc{unsafe}(\rho,\phi),\ \textsc{incomplete}(\rho),\\[-1mm]
&\textsc{safe-and-complete}(\chi)\}.
\end{aligned}
\]
Here $\chi$ is a registry-bound execution manifest,
\[
\begin{aligned}
\chi&=(v_s,k_r,v_r,\epsilon,t_i,t_d,\Psi,\mathbf h),\\[-1mm]
\Psi&=(s_0,\mathcal P\!\oplus\!\mathcal B,M),
\end{aligned}
\]
Here $v_s$ is the schema version; $(k_r,v_r)$ is the registry key and revision;
$\epsilon$ is the endpoint epoch; $[t_i,t_d]$ is the validity interval; and
$M$ is the registry-supplied access map. $\mathbf h$ contains canonical
component digests, including the initial-state projection. A validator checks
them against an independently pinned registry
before execution. Thus $\chi$ is a bound execution manifest, not a proof trace;
the verifier remains responsible for the trajectory proof.

\begin{algorithm}[t]
\caption{\repair{} template-guided certification}
\label{alg:verasync}
\begin{algorithmic}[1]
\Require $\mathcal P=(A,\prec_0,\Gamma,\mathcal D,\Phi,H)$,
$s_0,\mathcal T,K_{\max}$
\Ensure certified $(\mathcal P\oplus\mathcal B,\chi)$,
\textsc{unsupported}, or \textsc{iteration-limit}
\State $\mathcal B\gets\varnothing$
\For{$k=0,\ldots,K_{\max}$}
  \State $r\gets\Verify(\mathcal P\oplus\mathcal B,s_0)$
  \State \textbf{if} $r=\textsc{safe-and-complete}(\chi)$ \textbf{return}
  $(\mathcal P\oplus\mathcal B,\chi)$
  \State \textbf{if} $r=\textsc{incomplete}(\rho)$ \textbf{return}
  \textsc{unsupported}$(r)$
  \State \textbf{if} $k=K_{\max}$ \textbf{return}
  \textsc{iteration-limit}$(r)$
  \State $\eta\gets\operatorname{ExtractCore}(r.\mathrm{counterexample})$
  \State $\mathcal F\gets
  \operatorname{Compile}(\eta,\mathcal P,s_0,\mathcal T)\setminus\mathcal B$
  \State \textbf{if} $\mathcal F=\varnothing$ \textbf{return}
  \textsc{unsupported}$(\eta)$
  \State $\mathcal B\gets\mathcal B\cup\mathcal F$
\EndFor
\end{algorithmic}
\end{algorithm}

The loop instantiates a registered template for each supported safety
counterexample and rechecks the composition. It returns an
\textsc{incomplete} result as \textsc{unsupported}: the current templates add
precedence gates but do not add actions, relax request-ready times, or claim a
liveness repair.
Because $\mathcal B$ grows within finite $\mathfrak B_{\mathcal T}(\mathcal P)$, an
uncapped loop terminates within $|\mathfrak B_{\mathcal T}(\mathcal P)|$ iterations or
returns \textsc{unsupported}; soundness follows from the final whole-plan
check. Of 728 confirmation components, 712 certify after one repair iteration
and 16 after two; none consumes the additional guard depth $K_{\max}=4$.

The declared template library implements the source-evacuation dependency in
\eqref{eq:handover-group} and analogous scoped or versioned APPLY dependencies
for capacity transfer and PRB/slice reprofiling.

For the redundancy audit, $\mathcal B$ is \emph{1-minimal} iff it certifies
and $\mathcal B\setminus\{b\}$ does not for every $b\in\mathcal B$.
Re-certified single-edge deletion reaches this fixed point; the component
oracle evaluates a stronger lexicographic objective.
For a plan $\mathcal P$ and barriers $\mathcal B$, we report
\[
U(\mathcal P,\mathcal B)=
1-\frac{|\operatorname{TC}(\prec_0\cup\prec_{\mathcal B})|}
{\binom{|A|}{2}},
\]
where $\operatorname{TC}$ is the strict transitive closure over distinct
action pairs. Thus $U$ is the fraction left unordered. Summary values are
macro-averages over plans, so every plan has equal weight.

For the confirmation population, every pruned application-hook repair lies in
the direct APPLY-to-APPLY sublanguage. On access-map components of at most five actions,
the oracle enumerates every acyclic subset of that sublanguage and minimizes
direct barriers, ordered pairs, then worst-case completion time under the
original requests. \por{} screens candidates; exact search rechecks each
winner and the composition. The oracle comparison is therefore stated for the
confirmation's application-gate language; Lemma~\ref{lem:floor}
separately establishes the APPLY-to-REQUEST service-floor repair.

\subsection{Versioned execution and event-causal concretization}

Each verified dependency is released by an authoritative event and may guard
the dependent action's REQUEST or APPLY.

\begin{definition}[Admissible runtime concretization]
\label{def:concretization}
For evidence record $e$, let
\[
\begin{aligned}
\beta_b&:\{P_b\}\rightarrow\operatorname{Auth}(\Gamma_b),\\[-1mm]
\mathrm{Fresh}_{\chi}(e,s)&\Longleftrightarrow
t_i\leq t(e)\leq t_{\rm now}(s)\leq t_d\\[-1mm]
&\quad{}\land t_{\rm now}(s)-t(e)\leq
\operatorname{TTL}(\operatorname{type}(e))\\[-1mm]
&\quad{}\land\operatorname{epoch}(e)=\epsilon,\\[-1mm]
\mathrm{Match}_{\chi}(e,z,s)&\Longleftrightarrow
\operatorname{sig}(e)=\kappa_b(z)\land\mathrm{Fresh}_{\chi}(e,s),\\
\mathrm{Adm}_{\chi}(\beta_b,e,s)&\Longleftrightarrow
\beta_b(P_b)=z\in\operatorname{Auth}(\Gamma_b)\\[-1mm]
&{}\land\mathrm{Match}_{\chi}(e,z,s)\\[-1mm]
&{}\land\ell_b(z)=\mathrm{fired}.
\end{aligned}
\]
A gate $G_a$ for $B(b,a;P_b,G_a)$ opens only under
$\mathrm{Adm}_{\chi}(\beta_b,e,s)$. The binding
$\kappa_b(z)\equiv\kappa_b(P_b)$ fixes action, target scope, version, and
epoch while allowing the wire type of $z$ to differ from APPLY. The validator
accepts only a $\Gamma_b$-conformant runtime trace.
\end{definition}

\begin{theorem}[Contract-relative release-order refinement]
\label{thm:concretization}
Under Definition~\ref{def:concretization}, every runtime release of $G_a$
satisfies
\[
P_b\preceq_{\rho} z\prec_{\rho}G_a
\quad\text{and}\quad
t(P_b)\leq t(z)\leq t(G_a).
\]
Hence the concrete gate trace refines the verified release-order dependency
$P_b\prec G_a$. Missing, stale, wrong-scope, wrong-version, or ACCEPT-only
evidence cannot release $G_a$. Concretization may delay or block the plan, so
bounded completion remains an independent verification obligation.
\end{theorem}

\emph{Proof.} Opening $G_a$ identifies a fired $z=\beta_b(P_b)$.
Contract conformance gives $P_b\preceq_{\rho}z$, and the gate gives
$z\prec_\rho G_a$; failed validation and ACCEPT-only evidence remain closed.

The executor binds each barrier to a versioned APPLY callback or authorized
downstream evidence. Timeouts never count as APPLY; requests and evidence are
idempotent, restarts advance the certificate epoch, and all lifecycle
timestamps declared in $\Gamma_a$ remain reconstructible.

An APPLY gate also requires a registered application hook that can hold the
successor before state mutation. Otherwise the compiler lowers the dependency
and verifies the stronger executable order:
\[
P_b\prec P_a\ \mapsto\
\begin{cases}
P_b\prec P_a, & \operatorname{AppGate}(a),\\
P_b\prec R_a, & \text{otherwise}.
\end{cases}
\]
The lowered plan receives a new manifest only after full safety-and-completion
verification; the executor never treats request release as an APPLY gate.

\section{Scalable Certification with \por{}}
\label{sec:mipor}

\subsection{Exact reference and transparent projection}

The backend projects ACCEPT/COMPLETE only on a single-predecessor,
single-successor chain with independent finite delays and no branch, mutation,
semantic read, or gate; incident intervals combine by Minkowski sum. All other
cases use the six-event reference.

\begin{lemma}[Transparent lifecycle projection]
\label{lem:projection}
Contracting only transparent $A_a,C_a$ nodes preserves bounded safety labels
and the minimum and maximum occurrence time of every retained terminal event.
\end{lemma}
\emph{Proof sketch.} Trace contraction sums the two realized delays; every
composed delay decomposes into an admitted pair. The chain and no-read premises
preserve enabledness, labels, mutations, and terminal times both ways.

\subsection{RAN lifecycle and property footprints}

After Lemma~\ref{lem:projection}, \texttt{\_event\_access} maps retained
$R,D,P,O$ events to conservative semantic read/write tokens;
nontransparent $A_a$ or $C_a$ selects exact verification. The tokens cover
lifecycle status and predecessors; endpoint options; radio, load, PRB, and
slice state; evidence identity and freshness; terminal time, TTL, and horizon;
rollback snapshots; and repair-order closure.

Conflicts include read--write/write--write intersections, causality, barriers
and their closure, same-action events, and global rollback APPLY. A concurrent
APPLY also conflicts when a later rollback may read its pre-mutation snapshot.

\por{} runs in two stages. For
$\Omega=(\Phi,\mathcal B,\Gamma,\mathcal D,H)$, let
$\mathcal R_{\rm loc}(\Omega)$ contain every safety property, repair gate,
identity/freshness read, request-ready time, and per-action terminal read. For
action components $\mathcal K$, $\operatorname{Supp}(r)\subseteq A$ contains
the actions read by obligation $r$. Define
\[
\begin{aligned}
\mathsf{LocalityOK}(\mathcal K;\Omega)\Longleftrightarrow {}&
\forall r\in\mathcal R_{\rm loc}(\Omega),\ \exists K_i\in\mathcal K:\\[-1mm]
&\operatorname{Supp}(r)\subseteq K_i.
\end{aligned}
\]
The all-actions completion predicate and horizon are composed after component
exploration. Let $k_i\in\{\mathsf U,\mathsf I,\mathsf S\}$ be the verdict for
$r_i$; when all are $\mathsf S$, set
$T_{\min}=\max_iT_{i,\min}$ and $T_{\max}=\max_iT_{i,\max}$. Then
\[
C_\oplus(r_{1:m})=
\begin{cases}
\mathsf U, & \exists i:k_i=\mathsf U,\\
\mathsf I, & \exists i:k_i=\mathsf I,\\
\mathsf I, & (\forall i:k_i=\mathsf S)\land T_{\max}>H,\\
\mathsf S(T_{\min},T_{\max}), & \text{otherwise}.
\end{cases}
\]
Stage~1 requires $\mathsf{LocalityOK}$; a cross-component read disables it.
Stage~2 selects the lexicographically rooted persistent component at each
logical time. At state $s$, let $\tau=\tau(s)$ and use the exact-semantics
due set $E_\tau(s)=\{e:\mathrm{Enabled}_s(e)\}$.
$\mathcal Z_{s,\tau}$ adds zero-delay lifecycle successors, while
$\mathcal N_{s,\tau}$ adds events made ready by a plan or repair guard before
time advances. Both are conservative over endpoint and delay choices. Define
the future closure and lifted conflict together:
\[
\begin{aligned}
\mathcal F_{s,\tau}(e)&=\mu X.\bigl(\{e\}\cup\mathcal Z_{s,\tau}(X)
\cup\mathcal N_{s,\tau}(X)\bigr),\\[-1mm]
e\sim_{M,s}^{+}v&\Longleftrightarrow
\exists u\in\mathcal F_{s,\tau}(e),w\in\mathcal F_{s,\tau}(v):
u\sim_{M}w.
\end{aligned}
\]
Let $e\leadsto_{s,\tau} w$ mean that firing $e$ can make pending event $w$
ready before time advances. The registry admits Stage~2 only when
\[
\begin{aligned}
\mathsf{EnableSound}(\mathcal P,\mathcal B,M)
&\Longleftrightarrow {}\\[-1mm]
&\forall s\in\operatorname{Reach}_{\le H}(\mathcal P\oplus\mathcal B),\\[-1mm]
&\forall e\in E_{\tau(s)}(s),\ \forall w:\\[-1mm]
&e\leadsto_{s,\tau(s)}^{*}w\Rightarrow
w\in\mathcal F_{s,\tau(s)}(e).
\end{aligned}
\]
This is a semantic proof obligation; the implementation checks a finite,
conservative registry condition sufficient for it. Let
\[
\begin{aligned}
Q^0_{\Gamma}&=\{(u,v)\in F_a:d^-_{uv}=0\},\\[-1mm]
Q^0_G&=\{(u,v):u\Rightarrow_G^0v\},
\end{aligned}
\]
where $u\Rightarrow_G^0v$ denotes discharge of a same-tick plan or repair
guard, as derived from the registered DAGs, request-ready bounds, $\prec_0$,
and $\mathcal B$. Let $\widehat Q_Z$ and $\widehat Q_N$ be the corresponding
edges enumerated by $\mathcal Z$ and $\mathcal N$. Define
\begin{equation}
\begin{aligned}
\mathsf{EnableCover}(\mathcal P,\mathcal B,M)
\Longleftrightarrow {}&Q^0_{\Gamma}\subseteq\widehat Q_Z\\[-1mm]
&{}\land Q^0_G\subseteq\widehat Q_N.
\end{aligned}
\label{eq:enable-cover}
\end{equation}
The compiler also requires every rule that can change readiness without
advancing time to appear in $Q^0_{\Gamma}\cup Q^0_G$. These rules comprise
zero-delay DAG edges, plan and repair guards, and root request-ready
equalities. Write $\mathsf{StaticEnableOK}(\mathcal P,\mathcal B,M)$
for this registry check conjoined with Eq.~\eqref{eq:enable-cover}.

\begin{lemma}[Syntactic enablement coverage]
\label{lem:enable-cover}
If $\mathsf{StaticEnableOK}(\mathcal P,\mathcal B,M)$ holds, then
$\mathsf{EnableSound}(\mathcal P,\mathcal B,M)$ holds.
\end{lemma}
\emph{Proof sketch.} The registry premise makes
$Q^0_{\Gamma}\cup Q^0_G$ exhaustive. Equation~\eqref{eq:enable-cover} places
each immediate successor in $\mathcal Z$ or $\mathcal N$; induction places
every $e\leadsto^*_{s,\tau}w$ in $\mathcal F_{s,\tau}(e)$. Failed projection
or registry coverage selects exact verification.

$\operatorname{Options}(s,e,\tau(s))$ enumerates endpoint and equal-time
variants of firing $e$, together with schedule choices for successors that
$e$ newly makes ready. Events outside $C^\star$ stay due and commute past the
selected closure. The Stage~1 ablation branches over all enabled events with
the same state, options, and checks. Plans of at most eight actions use exact
search; larger plans use \por{} only after the declared checks pass.

\begin{table*}[t]
\centering
\caption{\textbf{Evidence coverage.} Plan certification establishes
cross-endpoint coverage; replay and the live E2 audit ground event bindings.}
\label{tab:endpoint-map}
\begin{tabular}{@{}L{0.15\textwidth}L{0.23\textwidth}L{0.29\textwidth}
                L{0.25\textwidth}@{}}
\toprule
Validation layer & Scope & Lifecycle evidence & Establishes \\
\midrule
Cross-endpoint certification &
144 plans: A1+E2 (48), A1+E2+O1 (80), E2+O1 (16) &
Per-action versioned APPLY/OBSERVE evidence under endpoint-specific
contracts and delay envelopes &
Plan-level safety and bounded completion; portable REQUEST-gate lowering and
registered application-gate repair \\
Mechanism-grounded replay &
E2 handover and quota; O1-bound cell sleep &
Serving-cell transition, scheduler version update, PHY sleep commit,
\texttt{HandoverEndOk}, and packet sinks &
Certificate-to-mechanism event binding and release behavior \\
Independent live contract audit &
OCUDU E2SM-RC Style~2/Action~6, 30 controls &
DU/MAC APPLY, executor COMPLETE, return ACK, and scheduler effect &
Live ACCEPT/APPLY/COMPLETE ordering on the audited OCUDU E2SM-RC path \\
\bottomrule
\end{tabular}
\end{table*}

\begin{algorithm}[t]
\caption{Two-stage \por{} verification}
\label{alg:mipor}
\begin{algorithmic}[1]
\Require $\mathcal P=(A,\prec_0,\Gamma,\mathcal D,\Phi,H),s_0,\mathcal B,
M$
\Ensure $k\in\{\mathsf U,\mathsf I,\mathsf S\}$; extrema if $k=\mathsf S$
\State $\Omega\gets(\Phi,\mathcal B,\Gamma,\mathcal D,H)$
\State $G_{\rm conf}\gets$ static action-conflict graph from
$M,\Gamma,\mathcal B,\Phi$
\If{not $\mathsf{StaticEnableOK}(\mathcal P,\mathcal B,M)$}
  \State \Return $\operatorname{Explore}
  (\mathcal P,s_0,\mathcal B,\textsc{Next}_{\rm exact})$
\EndIf
\State $\mathcal K\gets\operatorname{Components}(G_{\rm conf})$
\If{$|\mathcal K|>1$ and $\mathsf{LocalityOK}(\mathcal K;\Omega)$}
  \State $r_i\gets\operatorname{Explore}
  (\mathcal P[K_i],s_0,\mathcal B[K_i],\textsc{Next}_{\rm MI})$
  for each $K_i\in\mathcal K$
  \State \Return $\operatorname{Compose}(r_1,\ldots,r_{|\mathcal K|})$
  \Comment{Stage 1}
\EndIf
\State \Return $\operatorname{Explore}
(\mathcal P,s_0,\mathcal B,\textsc{Next}_{\rm MI})$
\Function{\textsc{Next}$_{\rm MI}$}{$s,E_{\tau}$}
  \If{$E_\tau=\varnothing$}
    \State \Return $\{\mathrm{Tick}\}$ if $\mathrm{Tick}(s)$, else
    $\varnothing$
  \EndIf
  \State $\mathcal F_{s,\tau}(e)\gets\operatorname{SameTickClosure}(s,e)$
  for each $e\in E_\tau$
  \State Build $G_{\tau}^{+}=(E_{\tau},\sim_{M}^{+})$, including
  barriers and rollback-visible snapshots
  \State $a^\star\gets\min\{\operatorname{action}(e):e\in E_{\tau}\}$
  \State $C^\star\gets$ component of $G_{\tau}^{+}$ containing $a^\star$
  \State \Return $\{(e,o):e\in C^\star,\;
  o\in\operatorname{Options}(s,e,\tau(s))\}$ \Comment{Stage 2}
\EndFunction
\end{algorithmic}
\end{algorithm}

\begin{definition}[Access-map soundness]
\label{def:access-sound}
For transition $t$, let
$\mathrm{RW}_{\rm sem}(t)=(R_{\rm sem}(t),W_{\rm sem}(t))$ and
$\mathrm{RW}_{M}(t)=(R_{M}(t),W_{M}(t))$, with
componentwise containment; $\gamma$ concretizes declared tokens. Then
\[
\mathsf{AccessSound}(M)\Longleftrightarrow
\forall t:\ \mathrm{RW}_{\rm sem}(t)
\subseteq\gamma(\mathrm{RW}_{M}(t)).
\]
$\mathrm{RW}_{\rm sem}$ includes enabledness, property, evidence/freshness,
state-key, and terminal-time accesses. An unassigned access disables
decomposition.
\end{definition}

The typed plan supplies lifecycle, gate, identity, freshness, and timing
tokens; a pinned registry supplies mechanism-specific resource accesses.
Missing token classes are rejected and unmapped accesses select exact search.
The theorem is registry-relative, and Section~\ref{sec:evaluation} exercises
that premise by access mutation and exact differential tests.

\begin{definition}[Implementation independence]
\label{def:independence}
For due transitions $u,v$, let $\mathrm{Conf}_s$ collect causal, same-action,
barrier, access-map, and rollback-snapshot conflicts. Write $u(s)$ for the
successor. The canonical key $\Xi(s)$ contains logical time; lifecycle status,
timestamps, and schedules; every radio/resource field read by a remaining
transition or property; evidence identity, version, epoch, and freshness;
terminal status; and rollback snapshots. Then
\[
\begin{aligned}
\mathrm{Conf}_s={}&\mathrm{Caus}\cup\mathrm{Same}\cup\mathrm{Bar}
\cup\mathrm{RW}_{M}\cup\mathrm{Snap}_s,\\[-1mm]
u\,\mathcal I_s\,v\Longleftrightarrow {}&
\neg\mathrm{Conf}_s(u,v)\\[-1mm]
&{}\land\mathrm{Enabled}_{u(s)}(v)\land\mathrm{Enabled}_{v(s)}(u)\\[-1mm]
&{}\land\Xi(v(u(s)))=\Xi(u(v(s))).
\end{aligned}
\]
\end{definition}

\begin{lemma}[AccessSound bridge]
\label{lem:access-sound}
If $\mathsf{AccessSound}(M)$ holds and the implemented conflict
detector establishes $\neg\mathrm{Conf}_s(u,v)$ for due transitions $u,v$,
then $u\,\mathcal I_s\,v$.
\end{lemma}

\emph{Proof.} Conservative access and rollback rules give mutual enabledness,
commutativity, and identical future-relevant state.

\begin{lemma}[Equal-time diamond]
\label{lem:diamond}
If $u\,\mathcal I_s\,v$, then $uv$ and $vu$ reach the same future-relevant state and
neither order changes whether a declared property is violated.
\end{lemma}

\emph{Proof.} This is immediate from Definition~\ref{def:independence}.

\begin{lemma}[Canonical representative]
\label{lem:canonical}
Every full equal-time execution is equivalent, by adjacent
$\mathcal I_s$-swaps, to one
whose next event belongs to the lexicographically rooted future-closed
component selected by \por{}.
\end{lemma}

\emph{Proof.} The future closure pulls every exposed conflict into $C^\star$;
all outside events commute and can be moved by Lemma~\ref{lem:diamond}.

\begin{theorem}[Declared-footprint bounded verdict preservation]
\label{thm:por}
For declared event contracts and properties, if
$\mathsf{AccessSound}(M)$ and
$\mathsf{EnableSound}(\mathcal P,\mathcal B,M)$ hold, and
Stage~1 is used only when $\mathsf{LocalityOK}(\mathcal K;\Omega)$ holds,
\por{} preserves the verifier verdict under the priority
$\mathsf U\succ\mathsf I\succ\mathsf S$. When the preserved verdict is
$\mathsf S$, it also preserves the minimum and maximum terminal completion
times. Transparent projection requires Lemma~\ref{lem:projection}.
\end{theorem}

\emph{Proof sketch.} Lemmas~\ref{lem:access-sound}--\ref{lem:canonical}
retain one representative per equal-time class with identical enabledness,
terminal status, state, and labels; TICK occurs only for an empty due set.
Under $\mathsf{LocalityOK}$, $C_\oplus$ returns $\mathsf U$ iff a component
does; absent $\mathsf U$, it returns $\mathsf I$ iff a component does or
$T_{\max}>H$. Otherwise
$T_{\rm comp}(\pi)=\max_iT_{{\rm comp},i}(\pi_i)$ preserves the extrema.
Lemma~\ref{lem:projection} handles transparent projection.

\FloatBarrier
\section{Implementation}

\setcounter{dbltopnumber}{1}
\begin{figure*}[t]
  \centering
  \includegraphics[width=0.99\textwidth]{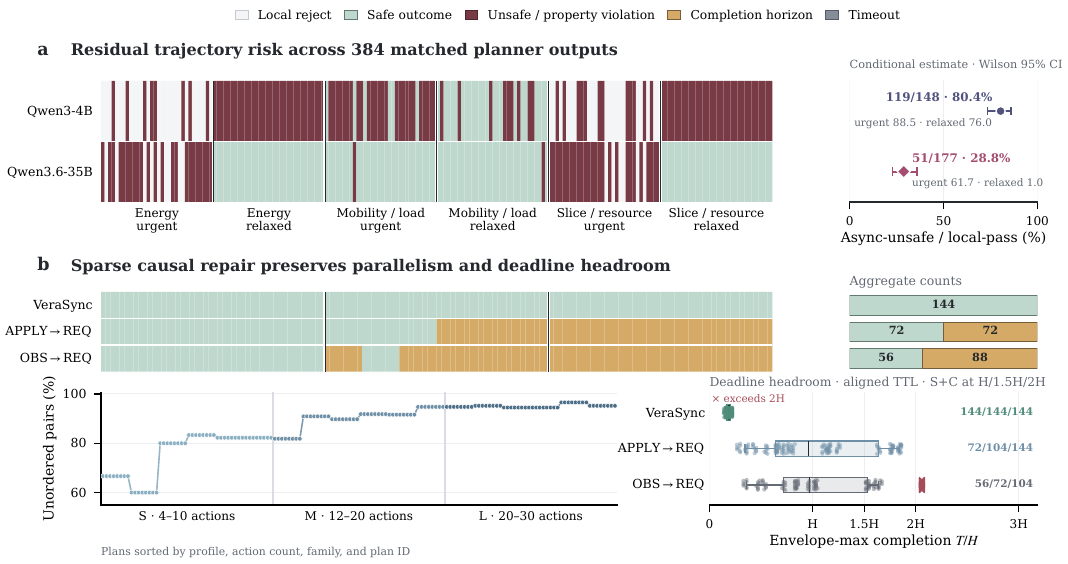}
  \caption{\textbf{Local validity leaves trajectory risk that sparse repair
  can close.} (a) Matched planner outputs with Wilson 95\% intervals.
  (b) Application-gate repair and release baselines on 144 confirmations;
  boxes show median/IQR, whiskers P5--P95, and crosses plans beyond $2H$.}
  \label{fig:discovery-repair}
\end{figure*}

\paragraph{Trusted computing base and certificate integrity}
The trusted computing base comprises the compiler and pinned registry,
verifier, certificate validator, and fail-closed executor. Together they bind
types, access maps, event contracts, properties, envelopes, identity fields,
deadlines, and component digests. Authenticated evidence producers and their
transport are part of this boundary. \repair{} only proposes barriers; the
verifier and validator must accept each proposal. Runtime evidence releases a
gate only when action, scope, version, and epoch match $\chi$.

Each of eight action types registers a bounded event DAG, evidence source,
RAN footprints, and same-tick enabling dependencies. Missing coverage selects
exact search. That engine retains nontransparent ACCEPT/COMPLETE nodes and
fail-closed gates; Lemma~\ref{lem:projection} permits the transparent
$R,D,P,O$ projection used by larger populations.

Mechanism replay sends APER controls through an xApp-like path to ns-O-RAN
commit \texttt{5f547ae8}~\cite{nsoran}. Handover APPLY is the serving-cell/RRC
transition~\cite{3gpp38331}; O1-bound cell sleep APPLY is a versioned PHY
commit~\cite{etsi2024o1}; quota APPLY is the versioned
\texttt{FdTbfqFfMacScheduler} update. Packet sinks expose the resulting
traffic effects, so APPLY is read from scheduler/PHY state rather than an
adapter-side mirror.

On OCUDU revision \texttt{9b0cfa6}, an xApp reaches a gNB
through SCTP/E2AP, O-RAN SC, and E2SM-RC Style~2/Action~6
~\cite{oran2023e2ap,oran2023e2smrc,ocuduproject}; instrumentation
records ACCEPT, slot-handler APPLY, executor COMPLETE, return ACK, and the
scheduler effect. Table~\ref{tab:endpoint-map} maps each evaluation layer to
the claim it supports.

\section{Evaluation}
\label{sec:evaluation}

\begin{figure*}[t]
  \centering
  \includegraphics[width=0.99\textwidth]{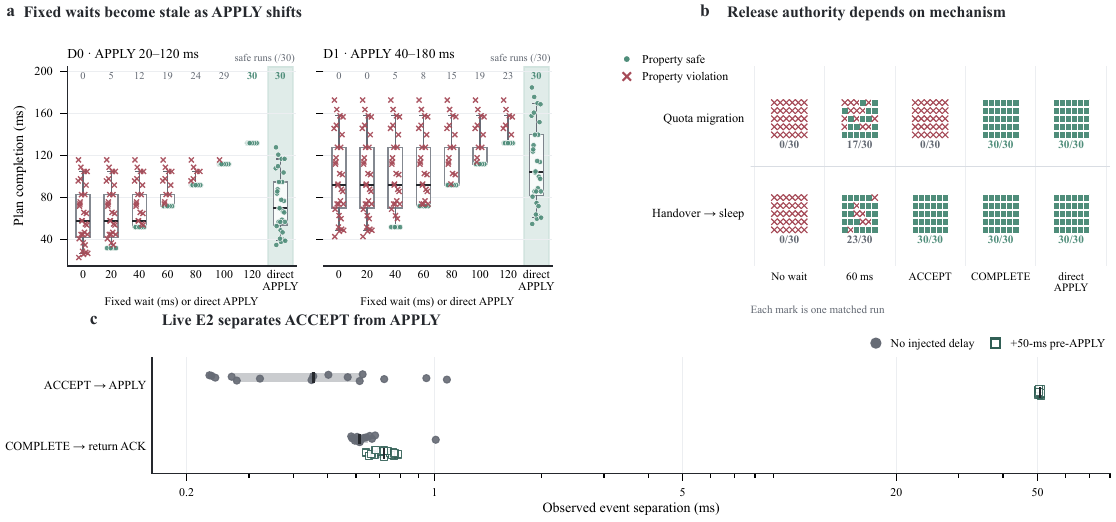}
  \caption{\textbf{Authoritative release adapts to actuation delay.}
  (a) Fixed-wait/direct-APPLY sweep (480 runs). (b) Mechanism-grounded release
  replay; each mark is one matched run (300 runs). (c) Live E2
  ACCEPT-to-APPLY separation (30 controls).}
  \label{fig:event}
\end{figure*}

\begin{figure*}[t]
  \centering
  \includegraphics[width=0.99\textwidth]{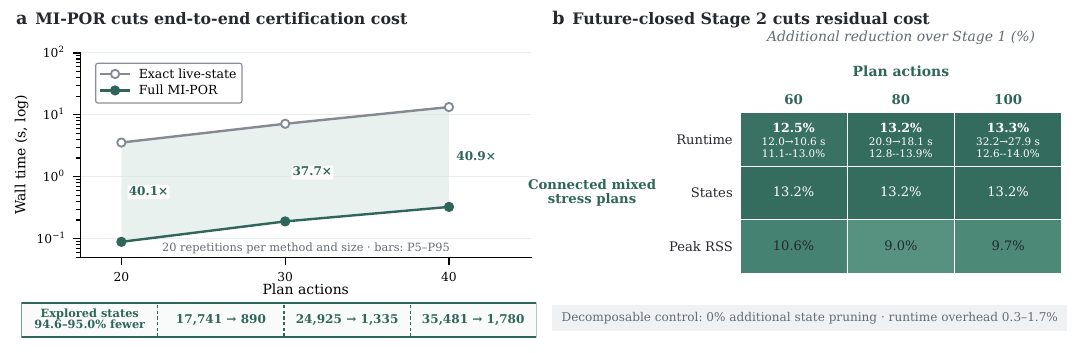}
  \caption{\textbf{\por{} reduces certification cost at both stages.}
  (a) Median wall time, P5--P95, and states (20 runs per size).
  (b) Residual Stage~2 savings; cells include absolute median seconds and
  P5--P95, with a decomposable control below.}
  \label{fig:mipor}
\end{figure*}

\subsection{Evidence strategy}

Before sampling, we fixed the generator, repair language, properties, delay
envelopes, and analysis protocol. The 144-plan confirmation spans six
families, 4/8/16 cells, and 4--30 actions, with one plan as the statistical
unit. The declared profile sets $H=64$ and evidence TTL $=64$. Joint
sensitivity evaluates $H\in\{64,96,128\}$ under both the declared TTL and a
horizon-aligned policy $\mathrm{TTL}=\max\{64,H\}$; Figure
~\ref{fig:discovery-repair}(b) reports the aligned policy. DELIVERY spans
0--1/1--2/1--3 ticks on E2/A1/O1; APPLY spans 1--2 except handover prepare at
1--3; OBSERVE spans 0--1 on E2 and 1--2 on A1/O1. These are symbolic stress
bounds. Figure~\ref{fig:event} provides millisecond-scale mechanism grounding.

\subsection{Residual risk after local validation}

The matched study assigns 192 base cases evenly to energy, mobility/load, and
slice/resource plans under urgent or relaxed timing. Horizon/deadline pairs
are 12/8 versus 24/20 ticks for energy and slice/resource, and 12/7 versus
20/17 for mobility. Both planners share state, goal, prompt, schema, compiler,
verifier, properties, and request seed; all calls compile.

A plan is local-pass when its schema, identifiers, action envelope, initial
state, request-time preconditions, and individual effects are valid under the
declared request order; this gate omits delay and interleaving exploration.
The deterministic planners are Qwen3.6-35B-A3B (3B active) and
Qwen3-4B-Instruct-2507~\cite{qwen36modelcard,qwen34bmodelcard}. Among
local-pass outputs, 51/177 35B plans are asynchronously unsafe (28.8\%, Wilson
95\% CI 22.6--35.9\%) versus 119/148 for 4B.

Every unsafe 35B plan starts and ends safely and is safe atomically; the 51
violations divide into 25 lifecycle-ordering and 26 cross-path resource cases.
They occur in 50/81 urgent versus 1/96 relaxed local-pass plans. Thus 28.8\%
is the conditional risk in this six-stratum benchmark, with timing pressure
providing the mechanism-level attribution.

\subsection{\repair{} closes plans with sparse ordering}

After fixing the generator and repair vocabulary, a pre-specified hash selects
eight of 16 eligible plans from each of 18 family-by-scale strata. Eligibility
requires compilation, satisfaction of the fixed local-pass criteria, and an
asynchronous counterexample before repair. Repair outcome and barrier count
are not selection inputs. The 144 confirmation plans have 4--30 actions; 96
combine mechanisms and 136 contain multiple unsafe cores.

The portable path lowers each synthesized APPLY-to-APPLY dependency to an
APPLY-to-REQUEST gate before pruning. Full verification certifies all 144
plans; none of its 1,400 barriers can be deleted, leaving 86.98\% of action
pairs unordered with median/P95 completion bounds of 17/17 ticks. When a
successor exposes a registered application hook, the direct APPLY-gate path
can retain the weaker symbolic order. Verified deletion then removes
140/1,400 barriers across 44 plans and raises the unordered fraction to
89.33\%, with median/P95 bounds of 12/14 ticks; every retained set is
1-minimal. Over 728 components of at most five actions, an oracle evaluates
42,888 acyclic barrier sets; the application-gate repair matches all three
lexicographic objectives on every plan within the declared APPLY-to-APPLY
vocabulary.

At $H=64$, APPLY- and OBSERVE-gated whole-plan serialization certify 72/144
and 56/144 plans. With $H=128$ and horizon-aligned TTL, they reach 144/144 and
104/144. Among plans certified at that setting, P95 completion bounds are 118
and 105 ticks, respectively (Fig.~\ref{fig:discovery-repair}(b)).

Boundary audits confirm fail-closed behavior: the API re-certifies 100/100
expressible controls and returns \textsc{unsupported} for 290/290 others;
ingestion releases 30/30 clean manifests once, rejects 240/240 static
mutations, and blocks dependent release for 90/90 runtime mismatches.

\subsection{Authoritative events realize the certificate}

Figure~\ref{fig:event}(b) uses 30 matched seeds per policy and mechanism.
No-wait release is always unsafe; 60 ms protects 17/30 quota and 23/30
handover--sleep runs. ACCEPT protects 0/30 quota but 30/30 handover runs,
where a fixed path delay places sleep APPLY 6.786 ms after the serving-cell
transition. This incidental protection does not make ACCEPT authoritative.
COMPLETE and direct APPLY protect every run, demonstrating that authority is
contract-specific rather than inferred from an ACK label.

In the 480-run delay-shift sweep, a 120-ms wait is safe in 30/30 nominal D0
runs but 23/30 stressed D1 runs; direct APPLY remains 30/30 and finishes 62 ms
sooner than that wait at the D0 median (Fig.~\ref{fig:event}(a)).

Across 30 OCUDU controls spanning 1/2/4 backlogged UEs and 0/50-ms pre-APPLY
delay, the xApp-to-DU/MAC path preserves
\[
R_a\prec_{\rho}A_a\prec_{\rho}P_a\prec_{\rho}C_a
\prec_{\rho}\mathrm{ACK}_a,
\qquad \theta_a(A_a)\leq\theta_a(P_a),
\]
ACCEPT and APPLY may share a timestamp but remain distinct events. Every
request has one matching APPLY, fresh version, and scheduler effect. Median
ACCEPT-to-APPLY separation is 0.455/50.897 ms; COMPLETE-to-return-ACK is
0.615/0.720 ms, placing both completion signals downstream of APPLY on this
contract. Intervals use the same host clock without overhead subtraction.

\subsection{\por{} preserves outcomes and reduces search}

\por{} matches exact verdicts, first violations, and completion extrema on all
130 property-stratified cases (78 unsafe, 26 safe, 26 repaired) and 1,120
directed or seeded-random cases, including 387 multi-component plans. The
six-event reference passes 56 nontransparent contracts and 256 valid
projections; hidden-dependency and access-class mutations trigger the expected
fallback or exact-versus-reduced witness.

On 20/30/40-action workloads, the full pipeline cuts states from
17,741/24,925/35,481 to 890/1,335/1,780 (94.6--95.0\%) and median time by
37.7--40.9$\times$; P95 latency is 90/191/327 ms. Stage~1 supplies the full
state reduction on these decomposable plans (Fig.~\ref{fig:mipor}(a)).

On connected 60/80/100-action workloads, Stage~2 reduces states by 13.2\%,
peak RSS by 9.0--10.6\%, and median runtime by 12.5--13.3\% beyond
decomposition. Median time falls from 12.0/20.9/32.2\,s to
10.6/18.1/27.9\,s; the decomposable control shows no added state pruning
(Fig.~\ref{fig:mipor}(b)).

\subsection{Deployment profiles}

\system{} certifies against registered event contracts, properties, and delay
envelopes. Missing access or same-tick coverage selects exact search; missing
or mismatched evidence closes the gate. A versioned mutation lease or endpoint
precondition protects the resource footprint, and each release rechecks
resource and registry versions, endpoint epoch, lease, and deadline. Any
mismatch requires re-certification. The 20--40-action profiles model online
plans; connected 60--100-action profiles model policy/configuration workflows.
New actions register event contracts plus access and enablement summaries.
Bounded completion assumes registered endpoints and authenticated producers
honor the registered request and lifecycle delay bounds in $\mathcal D$ and
$\{\delta_a\}_{a\in A}$; otherwise the gate stays closed until manifest expiry
at $t_d$.

\section{Related Work}

\textbf{RAN control and conflict handling.}
O-RAN distributes control across E2, A1, and O1
~\cite{polese2023understanding,oran2023e2ap,etsi2024a1,etsi2024o1}.
Existing systems characterize, mitigate, or block xApp conflicts
~\cite{prever2025pacifista,giannopoulos2025comix,adamczyk2023conflict,
wadud2024qacm,adamczyk2026accord,kwon2026orca,sultana2025conflict,
gelban2026outoforder,jawad2026copilot}.
\system{} certifies an already selected plan over bounded cross-endpoint
trajectories.

\textbf{Consistent updates and synchronization repair.}
Consistent updates schedule SDN changes, while NICE and SDNRacer expose
interleavings~\cite{reitblatt2011consistent,reitblatt2012abstractions,
jin2014dionysus,canini2012nice,elhassany2016sdnracer}. Synchronization synthesis
inserts dependencies over program or network transitions
~\cite{vechev2010synchronization,mcclurg2017syncnetwork}. These methods do not
generally bind a synthesized order to authoritative external lifecycle
evidence. \repair{} adds version- and scope-checked REQUEST gates or
application hooks, then re-certifies safety and bounded completion.

\textbf{Formal verification and runtime evidence.}
Formal O-RAN models verify energy and availability, while agentic model
checking verifies generated programs~\cite{metere2025oran,sun2026agentic}.
\system{} instead models action lifecycles, explores bounded trajectories
before actuation, and records which versioned event may release each gate.

\textbf{Partial-order reduction.}
\por{} builds on distributed event ordering~\cite{lamport1978time},
representative executions~\cite{peled1993representatives}, and
POR/DPOR/source sets~\cite{clarke,godefroid,musuvathi2008chess,
flanagan2005dpor,abdulla2014optimal,abdulla2017sourcesets}. It specializes
independence with lifecycle, resource, evidence, completion, repair, and
rollback footprints, including same-tick successors exposed by gates.

\FloatBarrier
\section{Conclusion}

\system{} certifies asynchronous RAN plans under per-action event contracts.
\repair{} compiles supported counterexamples into version- and scope-bound
gates and re-certifies the plan. All 144 confirmations certify while leaving
86.98\% of action pairs unordered with portable REQUEST gates and 89.33\%
with application hooks. Live E2 measurements ground release authority;
\por{} removes 94.6--95.0\% of states.

\bibliographystyle{IEEEtran}
\bibliography{references}

@IEEEtranBSTCTL{BSTcontrol,
  CTLuse_forced_etal       = {yes},
  CTLmax_names_forced_etal = {5},
  CTLnames_show_etal       = {5}
}

@inproceedings{nsoran,
  author       = {Andrea Lacava and Matteo Bordin and Michele Polese
                  and Rajarajan Sivaraj and Tommaso Zugno and Francesca Cuomo
                  and Tommaso Melodia},
  title        = {{ns-O-RAN}: Simulating {O-RAN} {5G} Systems in {ns-3}},
  booktitle    = {Proc. Workshop on ns-3},
  year         = {2023},
  pages        = {35--44},
  doi          = {10.1145/3592149.3592161}
}

@book{clarke,
  title        = {Model Checking},
  author       = {Clarke, Edmund M. and Grumberg, Orna and Peled, Doron A.},
  publisher    = {MIT Press},
  year         = {1999}
}

@book{godefroid,
  title        = {Partial-Order Methods for the Verification of Concurrent Systems},
  author       = {Godefroid, Patrice},
  publisher    = {Springer},
  year         = {1996}
}

@inproceedings{metere2025oran,
  title        = {Towards Achieving Energy Efficiency and Service Availability in {6G} {O-RAN} via Formal Verification},
  author       = {Metere, Roberto and Ye, Kangfeng and Gu, Yue and Zhang, Zhi and Alrajeh, Dalal and Sevegnani, Michele and Yadav, Poonam},
  booktitle    = {Proc. 12th Int. Symp. From Data to Models and Back
                  (DataMod 2024)},
  series       = {Lecture Notes in Computer Science},
  volume       = {15556},
  pages        = {137--157},
  year         = {2025},
  doi          = {10.1007/978-3-031-87908-1_9}
}

@misc{sun2026agentic,
  title        = {Agentic Model Checking},
  author       = {Sun, Youcheng and Liu, Jiawen and Kroening, Daniel and Xue, Jason},
  year         = {2026},
  eprint       = {2605.21434},
  archiveprefix = {arXiv},
  primaryclass = {cs.SE},
  note         = {arXiv:2605.21434}
}

@article{adamczyk2023conflict,
  title        = {Conflict Mitigation Framework and Conflict Detection in {O-RAN} {Near-RT} {RIC}},
  author       = {Adamczyk, Cezary and Kliks, Adrian},
  year         = {2023},
  journal      = {IEEE Communications Magazine},
  volume       = {61},
  number       = {12},
  pages        = {199--205},
  doi          = {10.1109/MCOM.018.2200752}
}

@inproceedings{sultana2025conflict,
  title        = {Experimental Evaluation of {xApp} Conflict Mitigation Framework in {O-RAN}: Insights from Testbed Deployment in {OTIC}},
  author       = {Sultana, Abida and Adamczyk, Cezary and Roy Chowdhury, Mayukh and Kliks, Adrian and Da Silva, Aloizio},
  booktitle    = {Proc. IEEE INFOCOM Workshops},
  year         = {2025},
  pages        = {1--6},
  doi          = {10.1109/INFOCOMWKSHPS65812.2025.11152996}
}

@article{giannopoulos2025comix,
  author       = {Anastasios E. Giannopoulos and Sotirios T. Spantideas
                  and George Levis and Alexandros S. Kalafatelis
                  and Panagiotis Trakadas},
  title        = {{COMIX}: Generalized Conflict Management in {O-RAN}
                  {xApps}---Architecture, Workflow, and a Power Control Case},
  journal      = {IEEE Access},
  year         = {2025},
  volume       = {13},
  pages        = {116684--116700},
  doi          = {10.1109/ACCESS.2025.3585774}
}

@article{prever2025pacifista,
  author       = {Pietro Brach del Prever and Salvatore D'Oro
                  and Leonardo Bonati and Michele Polese and Maria Tsampazi
                  and Heiko Lehmann and Tommaso Melodia},
  title        = {{PACIFISTA}: Conflict Evaluation and Management in Open {RAN}},
  journal      = {IEEE Transactions on Mobile Computing},
  year         = {2025},
  volume       = {24},
  number       = {10},
  pages        = {10590--10605},
  doi          = {10.1109/TMC.2025.3570632}
}

@inproceedings{kwon2026orca,
  author       = {Dae Cheol Kwon and Xinyu Zhang},
  title        = {Open {RAN} Conflict Agents: Detecting and Mitigating {xApp}
                  Conflicts with Generative Agents},
  booktitle    = {Proc. IEEE INFOCOM},
  year         = {2026},
  pages        = {1--10},
  doi          = {10.1109/INFOCOM59046.2026.11571236}
}

@inproceedings{adamczyk2026accord,
  author       = {Cezary Adamczyk and Adrian Kliks},
  title        = {{ACCoRD}: Actor-Critic Conflict Resolution with Deep Learning
                  for {O-RAN} {xApps}},
  booktitle    = {Proc. IEEE INFOCOM Workshops},
  year         = {2026},
  pages        = {1--6},
  doi          = {10.1109/INFOCOM59046.2026.11571309}
}

@article{polese2023understanding,
  author       = {Michele Polese and Leonardo Bonati and Salvatore D'Oro
                  and Stefano Basagni and Tommaso Melodia},
  title        = {Understanding {O-RAN}: Architecture, Interfaces, Algorithms,
                  Security, and Research Challenges},
  journal      = {IEEE Communications Surveys \& Tutorials},
  year         = {2023},
  volume       = {25},
  number       = {2},
  pages        = {1376--1411},
  doi          = {10.1109/COMST.2023.3239220}
}

@article{elkael2026agentran,
  author       = {Maxime Elkael and Salvatore D'Oro and Leonardo Bonati
                  and Michele Polese and Yunseong Lee and Koichiro Furueda
                  and Tommaso Melodia},
  title        = {{AgentRAN}: An Agentic {AI} Architecture for Autonomous
                  Control of Open {6G} Networks},
  journal      = {IEEE Communications Magazine},
  year         = {2026},
  pages        = {1--7},
  note         = {Early access},
  doi          = {10.1109/MCOM.001.2500563}
}

@article{wadud2024qacm,
  author       = {Abdul Wadud and Fatemeh Golpayegani and Nima Afraz},
  title        = {{QACM}: {QoS}-Aware {xApp} Conflict Mitigation in Open {RAN}},
  journal      = {IEEE Transactions on Green Communications and Networking},
  year         = {2024},
  volume       = {8},
  number       = {3},
  pages        = {978--993},
  doi          = {10.1109/TGCN.2024.3431945}
}

@article{gelban2026outoforder,
  author       = {Hams Gelban and Rouaa Naim and Ahmed Badawy},
  title        = {Detecting Out-of-Order Control Messages in {O-RAN}:
                  Dataset, Benchmarks, and Early-Warning Models},
  journal      = {IEEE Open Journal of the Communications Society},
  year         = {2026},
  volume       = {7},
  pages        = {1941--1957},
  doi          = {10.1109/OJCOMS.2026.3666940}
}

@article{jawad2026copilot,
  author       = {Md Asef Jawad and Mohammad Mehedi Hasan Munna
                  and Acramul Haque Kabir and Nahid Hossain Antu
                  and Reefka Fabliha Tulona},
  title        = {A Runtime Safety Copilot for {AI}-Native {O-RAN}:
                  Predictive Verification and Fail-Safe Enforcement in
                  {Near-RT} {RIC} Control Loops},
  journal      = {IEEE Access},
  year         = {2026},
  volume       = {14},
  pages        = {63106--63120},
  doi          = {10.1109/ACCESS.2026.3686132}
}

@inproceedings{reitblatt2012abstractions,
  author       = {Mark Reitblatt and Nate Foster and Jennifer Rexford
                  and Cole Schlesinger and David Walker},
  title        = {Abstractions for Network Update},
  booktitle    = {Proc. ACM SIGCOMM},
  year         = {2012},
  pages        = {323--334},
  doi          = {10.1145/2342356.2342427}
}

@inproceedings{vechev2010synchronization,
  author       = {Martin Vechev and Eran Yahav and Greta Yorsh},
  title        = {Abstraction-Guided Synthesis of Synchronization},
  booktitle    = {Proc. 37th ACM SIGPLAN-SIGACT POPL},
  year         = {2010},
  pages        = {327--338},
  doi          = {10.1145/1706299.1706338}
}

@inproceedings{flanagan2005dpor,
  author       = {Cormac Flanagan and Patrice Godefroid},
  title        = {Dynamic Partial-Order Reduction for Model Checking Software},
  booktitle    = {Proc. 32nd ACM SIGPLAN-SIGACT POPL},
  year         = {2005},
  pages        = {110--121},
  doi          = {10.1145/1040305.1040315}
}

@inproceedings{abdulla2014optimal,
  author       = {Parosh Aziz Abdulla and Stavros Aronis and Bengt Jonsson
                  and Konstantinos Sagonas},
  title        = {Optimal Dynamic Partial Order Reduction},
  booktitle    = {Proc. 41st ACM SIGPLAN-SIGACT POPL},
  year         = {2014},
  pages        = {373--384},
  doi          = {10.1145/2535838.2535845}
}

@article{lamport1978time,
  author       = {Leslie Lamport},
  title        = {Time, Clocks, and the Ordering of Events in a Distributed
                  System},
  journal      = {Communications of the ACM},
  year         = {1978},
  volume       = {21},
  number       = {7},
  pages        = {558--565},
  doi          = {10.1145/359545.359563}
}

@inproceedings{peled1993representatives,
  author       = {Doron Peled},
  title        = {All from One, One for All: On Model Checking Using
                  Representatives},
  booktitle    = {Proc. CAV},
  volume       = {697},
  pages        = {409--423},
  year         = {1993},
  doi          = {10.1007/3-540-56922-7_34}
}

@inproceedings{musuvathi2008chess,
  author       = {Madanlal Musuvathi and Shaz Qadeer and Thomas Ball
                  and Gerard Basler and Piramanayagam Arumuga Nainar
                  and Iulian Neamtiu},
  title        = {Finding and Reproducing Heisenbugs in Concurrent Programs},
  booktitle    = {Proc. 8th USENIX OSDI},
  pages        = {267--280},
  year         = {2008}
}

@article{abdulla2017sourcesets,
  author       = {Parosh Aziz Abdulla and Stavros Aronis and Bengt Jonsson
                  and Konstantinos Sagonas},
  title        = {Source Sets: A Foundation for Optimal Dynamic Partial Order
                  Reduction},
  journal      = {J. ACM},
  year         = {2017},
  volume       = {64},
  number       = {4},
  pages        = {25:1--25:49},
  doi          = {10.1145/3073408}
}

@inproceedings{reitblatt2011consistent,
  author       = {Mark Reitblatt and Nate Foster and Jennifer Rexford
                  and David Walker},
  title        = {Consistent Updates for Software-Defined Networks:
                  Change You Can Believe In!},
  booktitle    = {Proc. 10th ACM HotNets},
  year         = {2011},
  doi          = {10.1145/2070562.2070569}
}

@inproceedings{canini2012nice,
  author       = {Marco Canini and Daniele Venzano
                  and Peter Pere{\v s}{\'\i}ni and Dejan Kosti{\'c}
                  and Jennifer Rexford},
  title        = {A {NICE} Way to Test {OpenFlow} Applications},
  booktitle    = {Proc. 9th USENIX NSDI},
  pages        = {127--140},
  year         = {2012}
}

@inproceedings{jin2014dionysus,
  author       = {Xin Jin and Hongqiang Harry Liu and Rohan Gandhi
                  and Srikanth Kandula and Ratul Mahajan and Jennifer Rexford
                  and Roger Wattenhofer and Ming Zhang},
  title        = {Dionysus: Dynamic Scheduling of Network Updates},
  booktitle    = {Proc. ACM SIGCOMM},
  pages        = {539--550},
  year         = {2014},
  doi          = {10.1145/2619239.2626307}
}

@inproceedings{elhassany2016sdnracer,
  author       = {Ahmed El-Hassany and J{\'e}r{\'e}mie Miserez and Pavol Bielik
                  and Laurent Vanbever and Martin Vechev},
  title        = {{SDNRacer}: Concurrency Analysis for Software-Defined
                  Networks},
  booktitle    = {Proc. 37th ACM SIGPLAN PLDI},
  pages        = {402--415},
  year         = {2016},
  doi          = {10.1145/2908080.2908124}
}

@inproceedings{mcclurg2017syncnetwork,
  author       = {Jedidiah McClurg and Hossein Hojjat and Pavol {\v C}ern{\'y}},
  title        = {Synchronization Synthesis for Network Programs},
  booktitle    = {Proc. CAV},
  volume       = {10427},
  pages        = {301--321},
  year         = {2017},
  doi          = {10.1007/978-3-319-63390-9_16}
}

@article{bonati2021intelligence,
  author       = {Leonardo Bonati and Salvatore D'Oro and Michele Polese
                  and Stefano Basagni and Tommaso Melodia},
  title        = {Intelligence and Learning in {O-RAN} for Data-Driven {NextG}
                  Cellular Networks},
  journal      = {IEEE Communications Magazine},
  year         = {2021},
  volume       = {59},
  number       = {10},
  pages        = {21--27},
  doi          = {10.1109/MCOM.101.2001120}
}

@article{garciasaavedra2021oran,
  author       = {Andres Garcia-Saavedra and Xavier Costa-P{\'e}rez},
  title        = {{O-RAN}: Disrupting the Virtualized {RAN} Ecosystem},
  journal      = {IEEE Communications Standards Magazine},
  year         = {2021},
  volume       = {5},
  number       = {4},
  pages        = {96--103},
  doi          = {10.1109/MCOMSTD.101.2000014}
}

@article{hoffmann2024xapps,
  author       = {Marcin Hoffmann and Salim Janji and Adam Samorzewski
                  and Lukasz Kulacz and Cezary Adamczyk and Marcin Dryjanski
                  and Pawel Kryszkiewicz and Adrian Kliks and Hanna Bogucka},
  title        = {Open {RAN} {xApps} Design and Evaluation: Lessons Learnt and
                  Identified Challenges},
  journal      = {IEEE Journal on Selected Areas in Communications},
  year         = {2024},
  volume       = {42},
  number       = {2},
  pages        = {473--486},
  doi          = {10.1109/JSAC.2023.3336190}
}

@article{oh2011dynamic,
  author       = {Eunsung Oh and Bhaskar Krishnamachari and Xin Liu
                  and Zhisheng Niu},
  title        = {Toward Dynamic Energy-Efficient Operation of Cellular
                  Network Infrastructure},
  journal      = {IEEE Communications Magazine},
  year         = {2011},
  volume       = {49},
  number       = {6},
  pages        = {56--61},
  doi          = {10.1109/MCOM.2011.5783985}
}

@article{li2011energyefficient,
  author       = {Geoffrey Ye Li and Zhikun Xu and Cong Xiong
                  and Chenyang Yang and Shunqing Zhang and Yan Chen
                  and Shugong Xu},
  title        = {Energy-Efficient Wireless Communications: Tutorial, Survey,
                  and Open Issues},
  journal      = {IEEE Wireless Communications},
  year         = {2011},
  volume       = {18},
  number       = {6},
  pages        = {28--35},
  doi          = {10.1109/MWC.2011.6108331}
}

@article{caballero2017multitenant,
  author       = {Pablo Caballero and Albert Banchs and Gustavo de Veciana
                  and Xavier Costa-P{\'e}rez},
  title        = {Multi-Tenant Radio Access Network Slicing: Statistical
                  Multiplexing of Spatial Loads},
  journal      = {IEEE/ACM Transactions on Networking},
  year         = {2017},
  volume       = {25},
  number       = {5},
  pages        = {3044--3058},
  doi          = {10.1109/TNET.2017.2720668}
}

@techreport{oran2023e2ap,
  author       = {{O-RAN Alliance}},
  title        = {{E2} Interface: Application Protocol ({E2AP})},
  institution  = {O-RAN Alliance},
  number       = {O-RAN.WG3.E2AP-R003-v03.00},
  year         = {2023}
}

@techreport{oran2023e2smrc,
  author       = {{O-RAN Alliance}},
  title        = {{E2} Service Model ({E2SM}) {RAN} Control},
  institution  = {O-RAN Alliance},
  number       = {O-RAN.WG3.E2SM-RC-R003-v03.00},
  year         = {2023}
}

@techreport{etsi2024a1,
  author       = {{ETSI}},
  title        = {{A1} Interface: General Aspects and Principles},
  institution  = {European Telecommunications Standards Institute},
  number       = {ETSI TS 103 983 V3.1.0},
  year         = {2024},
  month        = jan,
  note         = {O-RAN.WG2.A1GAP-R003-v03.01}
}

@techreport{etsi2024o1,
  author       = {{ETSI}},
  title        = {{O-RAN} Operations and Maintenance Interface Specification},
  institution  = {European Telecommunications Standards Institute},
  number       = {ETSI TS 104 043 V11.0.0},
  year         = {2024},
  month        = jun,
  note         = {O-RAN.WG10.O1-Interface-R003-v11.00}
}

@techreport{3gpp38331,
  author       = {{ETSI}},
  title        = {{5G}; {NR}; Radio Resource Control ({RRC}); Protocol Specification},
  institution  = {European Telecommunications Standards Institute},
  number       = {ETSI TS 138 331 V18.6.0},
  year         = {2025},
  month        = jul,
  note         = {3GPP TS 38.331 Release 18}
}

@misc{qwen36modelcard,
  author       = {{Qwen Team}},
  title        = {{Qwen3.6-35B-A3B} Model Card},
  year         = {2026},
  howpublished = {\url{https://huggingface.co/Qwen/Qwen3.6-35B-A3B}},
  note         = {Experimental revision 1a5ae24e867f8d82388070d3f61590158a01d15c;
                  accessed July 2026}
}

@misc{qwen34bmodelcard,
  author       = {{Qwen Team}},
  title        = {{Qwen3-4B-Instruct-2507} Model Card},
  year         = {2025},
  howpublished = {\url{https://huggingface.co/Qwen/Qwen3-4B-Instruct-2507}},
  note         = {Revision cdbee75f17c01a7cc42f958dc650907174af0554;
                  accessed July 2026}
}

@misc{ocuduproject,
  author       = {{OCUDU Project}},
  title        = {{OCUDU}: Open-Source {5G} and Beyond {CU/DU}},
  year         = {2026},
  howpublished = {\url{https://gitlab.com/ocudu/ocudu}},
  note         = {Revision 9b0cfa600d9d693bf56277e4575dc8c4b8b729bb;
                  accessed July 2026}
}

\end{document}